\documentclass[conference]{IEEEtran}
\IEEEoverridecommandlockouts

\AtBeginDocument{%
  \providecommand\BibTeX{{%
    Bib\TeX}}}

\usepackage{tikz}

\newcommand\copyrighttext{%
  \footnotesize \textcopyright 2026 IEEE.  Personal use of this material is permitted.  Permission from IEEE must be obtained for all other uses, in any current or future media, including reprinting/republishing this material for advertising or promotional purposes, creating new collective works, for resale or redistribution to servers or lists, or reuse of any copyrighted component of this work in other works.}

\newcommand\copyrightnotice{%
\begin{tikzpicture}[remember picture,overlay]
\node[anchor=south,yshift=20pt] at (current page.south) {\parbox{\textwidth}{ \copyrighttext}};
\end{tikzpicture}%
}

\usepackage{amsmath,amsfonts,amsthm}
\usepackage{subcaption}
\usepackage{graphicx}
\usepackage{textcomp}
\usepackage{xcolor}
\usepackage{twoopt}
\usepackage{enumitem}
\usepackage{algorithm}
\usepackage{algorithmic}
\usepackage{amsmath}
\usepackage{mathtools}
\mathtoolsset{showonlyrefs}

\usepackage{hyperref}
\def\BibTeX{{\rm B\kern-.05em{\sc i\kern-.025em b}\kern-.08em
    T\kern-.1667em\lower.7ex\hbox{E}\kern-.125emX}}

\usepackage{bm,bbm}
\usepackage{xcolor}
\newtheorem{theorem}{Theorem}

\newtheorem{corollary}{Corollary}
\newtheorem{remark}{Remark}
\allowdisplaybreaks

\newcommand{\E}{\mathbbm{E}}
\newcommand{\p}{\mathbbm{P}}
\newcommand{\R}{\mathbbm{R}}

\makeatletter
\newcommand{\linebreakand}{%
  \end{@IEEEauthorhalign}
  \hfill\mbox{}\par
  \mbox{}\hfill\begin{@IEEEauthorhalign}
}
\makeatother

\begin{document}

\title{Opportunistic Scheduling for Optimal Spot Instance Savings in the Cloud\thanks{*while author was working at Nokia Bell Labs, Murray Hill, NJ, USA.}}

\author{
\IEEEauthorblockN{Neelkamal Bhuyan*}
\IEEEauthorblockA{\textit{Georgia Institute of Technology} \\
Atlanta, GA, USA \\
nbhuyan3@gatech.edu}
\and
\IEEEauthorblockN{Randeep Bhatia}
\IEEEauthorblockA{\textit{Nokia Bell Labs}\\
Murray Hill, NJ, USA \\
randeep.bhatia@nokia-bell-labs.com}
\linebreakand
\IEEEauthorblockN{Murali Kodialam}
\IEEEauthorblockA{\textit{Nokia Bell Labs}\\
Murray Hill, NJ, USA \\
murali.kodialam@nokia-bell-labs.com}
\and
\IEEEauthorblockN{TV Lakshman}
\IEEEauthorblockA{\textit{Nokia Bell Labs}\\
Murray Hill, NJ, USA \\
tv.lakshman@nokia-bell-labs.com}
}

\maketitle
\copyrightnotice

\begin{abstract}
We study the problem of scheduling delay-sensitive jobs over spot and on-demand cloud instances to minimize average cost while meeting an average delay constraint. Jobs arrive as a general stochastic process, and incur different costs based on the instance type. This work provides the first analytical treatment of this problem using tools from queuing theory, stochastic processes, and optimization. We derive cost expressions for general policies, prove queue length one is optimal for low target delays, and characterize the optimal wait-time distribution. For high target delays,  we identify a knapsack structure and design a scheduling policy that exploits it. An adaptive algorithm is proposed to fully utilize the allowed delay, and empirical results confirm its near-optimality.
\end{abstract}

\maketitle

\section{Introduction}

Minimizing cloud operational costs is a growing concern, especially with the rapid adoption of Large Language Models (LLMs) and AI tools. Among the major contributors to cloud expenses are compute instances. While large enterprises can afford on-demand instances, smaller organizations rely on budget-friendly alternatives—chiefly spot instances—which are cheaper but less reliable in availability.

To offload unsold capacity, cloud providers introduced pricing models such as Amazon’s spot instance system~\cite{AWS_EC2},  Google  Preemptible VMs~\cite{GCP} and Oracle~\cite{oracle_spots}. These spot instances are significantly cheaper (by $3$--$10\times$~\cite{295489}) than on-demand ones, but come with the risk: of uncertain preemption, which is problematic for delay-sensitive edge applications \cite{11181549,bhuyan2022multi,bhuyan2022provable} or low-latency financial trading~\cite{numerai,mosaic}, where timely computation is critical.

Prior works~\cite{chetto2014optimal},~\cite{islam2020scheduling},~\cite{ghor2011real},~\cite{295489},~\cite{9640599},~\cite{taghavi2023cost} have not addressed this problem through a rigorous job-scheduling framework. \cite{295489} addresses this delay-constrained scheduling challenge only for a single job. Their \textit{Uniform Progress} policy ensures consistent job progress but neither addresses a stream of jobs nor establishes optimality guarantees. In contrast, we study the problem of \textit{optimally scheduling a continuous stream of delay-sensitive jobs} over spot and on-demand instances, aiming to minimize the average cost per job while ensuring the average delay does not exceed a threshold $\delta$. 

\begin{figure}[htb]
  \centering
  \includegraphics[width=0.8\linewidth]{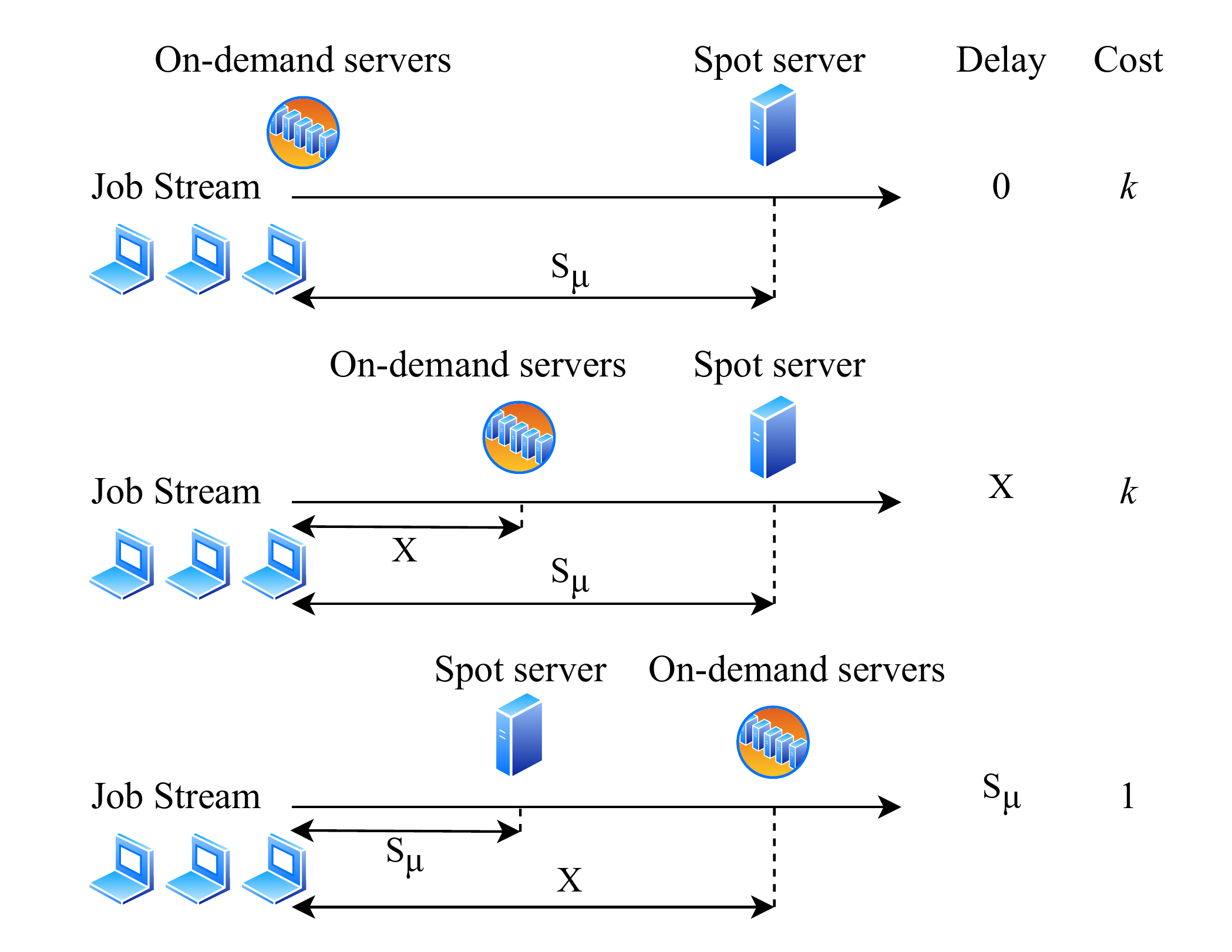}
  \caption{Job Arrivals and Departures: Delay and Costs}
  \label{spot_cases}
  \vspace{-5pt}
\end{figure}

\textbf{\textit{Main Contributions:}} In this work, we formalize the optimization plus scheduling problem from a distribution-agnostic $G/G/1$ queuing perspective, identifying two regimes: (i) tight delay constraint and, (ii) relaxed delay constraint. Closed-form tractability in the strong delay  regime allows us to establish an optimal queue length of \textit{one} (Theorem \ref{thm: gen_cost} and \ref{thm: opt_queue}) and formulate the optimal greedy wait-time for spot instances (Theorem \ref{thm: opt_prob}). 

Lack of closed-form solutions for a $G/G/1$ system make the relaxed delay regime significantly harder to analyze. We circumvent this by identifying a novel \textit{knapsack structure} in the greedy heuristic (Theorem \ref{thm: gen_sys_greedy}), that makes it near-optimal. It leads to a 3-phase structure, involving a fractional admission policy \eqref{eqn: r_formula} that needs to be learned.

Inspired from recent advances in learning-augmented online control \cite{bhuyan2025estimate,9483309,bhuyan2026scale}, we design an Adaptive Admission Control Policy (Algorithm \ref{alg: adaptive_adm_ctrl}) that learns the optimal fractional admission while scheduling jobs on spot instances. Our experiments under various spot availability patterns (Section \ref{section:experiments}) show fast convergence to the optimal policy (Fig. \ref{fig: gcp_low_delta} - \ref{fig: mm1 large delta}).

\vspace{-10pt}

\section{Model and Preliminaries}
\subsection{Analytical model}
We consider the problem of scheduling an incoming \textit{stream} of jobs either on a spot instance or an on-demand instance where the cost of using a spot server is $1$ and that of using an on-demand server is $k$. Throughout, we use the terms ``instance'' and ``server'' interchangeably. According to traces generated in \cite{295489,azure_spots}, $k$ usually lies in the range $3$ to $10$ for GPUs and CPUs. Figure \ref{spot_cases} below illustrates the general system. 

\begin{table}
    \centering
    \begin{tabular}{|c|c|}
    \hline
    \textbf{Notation}     & \textbf{Entity}\\
    \hline
    $k$          & on-demand instance's cost per job\\
    $\delta$ & maximum allowed expected delay\\
    $C$ & cost per job \\
    $N_A(t)$     & number of job arrivals in time $t$\\
    $N_{S_\mu}(t)$     & number of spot instance arrivals in time $t$\\
    $A$     & job inter-arrival time (IID)\\
    $S_\mu$     & spot inter-arrival time (IID)\\
    $\lambda$ & $1/\E[A]$\\
    $\mu$ & $1/\E[S_\mu]$ \\
    $X$     & a job's maximal wait time (IID)\\
    $T$     & time an incoming job spends in the system\\
    $N$     & maximum queue length, if enforced\\
    $\pi_i$     & steady state probability of queue length being $i$\\
    $\mathcal{L}\{h\}(s)$ & Laplace transform of $h(t)$ evaluated at $s$\\
    \hline
    \end{tabular}
    \caption{Notations}
    \label{tab: notations}
\end{table}

In related domains, such as scheduling \cite{Koutsopoulos2011,11038902}, caching \cite{chae2016caching,sadeghi2018optimal,madnaik2022renting} and online control \cite{9483309,bhuyan2024best}, stochastic modeling has been used to capture availability dynamics. We, therefore, model the sequence of incoming jobs $N_A(t)$ as a general stochastic renewal process, where the inter-arrival times of the jobs $\{A_n\}_{n=1}^\infty$ are independent and identically distributed.

The service mechanism is as follows. The spot instances' presence is modeled as an stochastic arrival process $N_{S_\mu}(t)$, with inter-arrival times between spot arrivals, $S_{\mu}^n$, being IID (and occasionally written without $n$ in the superscript). When a spot instance arrives into the system and finds a job waiting, it picks the first one in the queue for spot-service (FCFS) and serves it. A job can be processed any time by an on-demand instance, albeit at a higher cost. In our setting, we model the use of these in the form of jobs leaving the queue without service from spot instance. From the job's perspective, if it is assigned to any instance (spot or on-demand), it is considered finished for the dispatch algorithm. Existence of such spot instances, where job ``interruption ratio'' is low, is common knowledge \cite{AWS_EC2,azure_spots,10.1145/3589334.3645548}.

The choice between waiting for a spot instance and leaving for on-demand service arises from jobs being delay-sensitive. We model this as the job's \textit{average waiting time} in the system (both on-demand and spot) being at most $\delta >0$. On-demand instances have zero wait-time while spot instances can have possibly infinite wait-time. On the other hand, spot instances are significantly cheaper than on-demand.

So, lets consider the options a job has when it arrives. The first choice is whether to join the queue for a spot instance at all. The probability of this event is denoted as $q_n$, where $n$ is the number of jobs in the system. Once, it joins the queue, it has to make a choice between waiting in the system or leaving for on-demand service.  It decides a maximal waiting time for service, modeled as the random variable $X_n \geq 0$ (IID, and occasionally written without the subscript). The job will either get serviced by a spot instance before $X_n$ time expires (after its arrival) or choose an on-demand server.

We study the above described system in the steady state. We denote the cost per job with the random variable $C$ and the wait-time experienced by any job as $T$. We denote the class of scheduling policies by $\mathcal{P} := \left\{\left(\{q_n\}_{n=0}^\infty,\{X_n\}_{n=1}^\infty\right) : q_n \in [0,1], X_n \geq0 \right\}$. The goal of the scheduler is to solve the following optimization problem:
\begin{align}
    \min_{\substack{\left(\{q_n\}_{n=0}^\infty,\{X_n\}_{n=1}^\infty\right) \in \mathcal{P}\\ \E[T] \leq \delta}} \E[C]
\end{align}
where $\E[\cdot]$ denotes expectation with respect to the randomness in the system. The system is referred using Kendall's notation for queuing, which is written as $\cdot / \cdot / \cdot + (\cdot)$. 

\subsection{Cost per Job for a General System}
We start with a general system where jobs and spot instances arrive as renewal processes.
\begin{theorem}\label{thm: gen_cost}
Consider a general job arrival process and a general spot arrival process $(G/G/1)$. The cost per job for any scheduling policy in steady state is
\begin{align*}
    \E[C] = k - (k-1)\cdot \frac{\E[A]}{\E[S_\mu]} (1-\pi_0)
\end{align*}
where $\E[A]: = \frac{1}{\lambda}$ is the average job inter-arrival time, $\E[S_\mu]: = \frac{1}{\mu}$ is the average spot inter-arrival time and $\pi_0$ is the steady state probability of the queue being empty.
\end{theorem}

The above result highlights the key driving force int the job scheduling problem in dual-processor systems (spot and on-demand): the sole dependence on the long term probability of the queue being empty., which is in fact intuitive. From the spot instance's perspective, it saves a job from incurring a cost of $k$, and that can only happen if there are jobs in the system, the probability of which is $(1-\pi_0)$. What is interesting is this holds under virtually no assumptions about the system. Below, we detail the proof of Theorem \ref{thm: gen_cost}.
\vspace{-5pt}
\begin{proof}
Let $Q_n$ be the event where the $n^{th}$ spot instance serves a job. We look at the running cost per job,
\begin{align}
    C(t) &=  \frac{k\cdot N_A(t) - (k-1)\cdot \sum_{n=1}^{N_{S_\mu}(t)} \mathbbm{1}(Q_n)}{N_A(t)}\\
    &= k - (k-1)\frac{t}{N_A(t)} \frac{N_{S_\mu}(t)}{t} \frac{\sum_{n=1}^{N_{S_\mu}(t)} \mathbbm{1}(Q_n)}{N_{S_\mu}(t)}
\end{align}
where $N_A(t)$ is the number of jobs arrived in the system and $N_{S_\mu}(t)$ is the number of spot instances arrived in the system. Now, remember that $C(t) \in [1,k]$, meaning $C(t)$ for a given $t$ is an integrable random variable. So,
\begin{align}
    \lim_{t \to \infty} C(t) =& \text{ } k - (k-1)\left( \lim_{t \to \infty} \frac{t}{N_A(t)}\right) \times \ldots\\
    & \ldots \times \left( \lim_{t \to \infty} \frac{N_{S_\mu}(t)}{t}  \right) \lim_{t \to \infty} \frac{\sum_{n=1}^{N_{S_\mu}(t)} \mathbbm{1}(Q_n)}{N_{S_\mu}(t)}\\
    =& \text{ } k - (k-1)\cdot \frac{\E[A]}{\E[S_\mu]} \lim_{N \to \infty} \frac{\sum_{n=1}^{N} \mathbbm{1}(Q_n)}{N}.
\end{align}
where we use the Renewal Theorem and the fact that (for an random variable) $Y \to x_0$ $a.s \implies$ $\frac{1}{Y} \to \frac{1}{x_0}$ $a.s$. Now, taking expectation on both sides
\begin{align}
\begin{split}
    \E\left[\lim_{t \to \infty} C(t) \right] =& k - (k-1) \frac{\E[A]}{\E[S_\mu]} \cdot \E\left[ \lim_{N \to \infty} \frac{\sum_{n=1}^{N} \mathbbm{1}(Q_n)}{N}  \right]
\end{split}
\end{align}
Since, $C(t) \leq k$ and $\frac{\sum_{n=1}^{N} \mathbbm{1}(Q_n)}{N} \leq 1$, dominated convergence theorem allows switching expectation and limit to get
\begin{align}
    \lim_{t \to \infty} \E[C(t)] &= k - (k-1)\cdot \frac{\E[A]}{\E[S_\mu]} \lim_{N \to \infty} \frac{\sum_{n=1}^{N} \p(Q_n)}{N}\\
    &= k - (k-1)\cdot \frac{\E[A]}{\E[S_\mu]} \lim_{N \to \infty} \frac{\sum_{n=1}^{N} \left(1-p^{(n)}_0 \right)}{N}\\
    &= k - (k-1)\cdot \frac{\E[A]}{\E[S_\mu]} (1-\pi_0)
\end{align}
where $p_0^{(n)}$ is the transient probability of the queue-length being zero and $\pi_0$ being the steady state probability.
\end{proof}



\section{Optimal scheduling under low target delay $\delta$}
In the previous subsection, we established the cost per job for any system. Now, we look at the optimal queue length to allow when $\delta$ is small.
\subsection{Optimal queue-length}
\begin{theorem}\label{thm: opt_queue}
	For $\delta \leq \frac{\p(A\leq S_\mu)}{\lambda}$, the optimal queue length is \textbf{one} and the average cost per job will satisfy
	\begin{align*}
		\E[C] = k - (k-1)\mu \delta
	\end{align*}
	for the appropriate maximal wait-time distribution $f_X(\cdot)$.
\end{theorem}


Under tight average deadlines, this result establishes the optimal queue length for any dual-service cloud compute system. This is crucial for light tail distributions, where spot instances are expected to arrive quicker, but the tight deadline makes it sub-optimal for jobs further in the queue to suffer cascading delays, as illustrated by the proof.
\begin{proof}
    Consider a system where the queue-length is capped to one. Further, if a job joins the queue for spot-service, it waits indefinitely until a spot instance arrives. For such a system in steady-state, an incoming job will observe a queue-length of $1$ (the previous job in the queue) when no spot-instance has arrived during the job inter-arrival duration. This can only happen when the spot inter-arrival times are longer than the job inter-arrival times. Therefore, the probability of finding a queue-length of $1$ here is:
    \begin{align}
        \pi_1 = \p(A\leq S_\mu)
    \end{align}
    For such a system, the average wait-time is:
    \begin{align}
        \E[T^*] &= \frac{1}{\lambda}\left(0\cdot \pi_0 + 1\cdot\pi_1\right) = \frac{\p(A\leq S_\mu)}{\lambda}
    \end{align}
    Recall that we considered indefinite wait-time for the job that joins the queue. Consider the same single-slot system now with wait-time $X$ having distribution $F_X: [0,\infty)\to [0,1]$.
    For any $\delta \in \left[0, \frac{\p(A\leq S_\mu)}{\lambda} \right]$, there exists a wait-time $X$ with distribution $F_X$ such that the average wait-time of the corresponding single-slot system $\E[T] = \delta$.

    For any general system, the goal is to minimize the average cost per job in Theorem \ref{thm: gen_cost}, which amounts to maximizing $(1-\pi_0) = \sum_{n=1}^\infty \pi_n.$ This maximization has to be done under the \textit{binding} constraint $$\frac{1}{\lambda} \sum_{n=1}^\infty n \pi_n = \E[T] = \delta.$$ Now, for $\delta \in \left[0, \frac{\p(A\leq S_\mu)}{\lambda} \right]$, consider two scheduling policies: (a) single-slot (queue-length at most one) and, (b) allowing queue-length $N>1$. Optimizing over the wait-time distribution in both cases will involve utilizing the given constraint $\delta$ to the brim. This means:
    \begin{align}\label{eqn15}
        \frac{1}{\lambda}\sum_{n=1}^N \pi_n^{(b)} < \frac{1}{\lambda}\sum_{n=1}^N n \pi_n^{(b)} = \delta = \frac{1}{\lambda}\sum_{n=1}^1 n \pi_n^{(a)} = \frac{1}{\lambda}\sum_{n=1}^1 \pi_n^{(a)}
    \end{align}
    Therefore, the cost of system (a) will be smaller than cost of system (b), proving that single-slot queue is optimal for strong delay constraint, namely, $\delta \in \left[0, \frac{\p(A\leq S_\mu)}{\lambda} \right].$

Under this strong delay constraint setting, the optimal cost can be formulated by applying Theorem \ref{thm: gen_cost} and the fact that $(1-\pi_0^{(a)}) = \lambda \delta$ in \eqref{eqn15} to get
\begin{align}
    \E[C] = k - (k-1)\mu \delta.
\end{align}

\end{proof}
\vspace{-5pt}
However, simply capping queue length to one does not give the optimal cost. One also has to choose the correct maximal wait-time distribution $F_X$ for the jobs that join the queue. The next subsections deal with this problem in extensive detail.

\subsection{Optimal wait-time distribution}
Here, we restrict ourselves to single-slot queue as it is now well-established that it is the optimal policy for $\delta \leq \frac{\p(A\leq S_\mu)}{\lambda}$. We again consider a general job arrival process and spot arrival process. Our next result describes the optimal maximal wait-time distribution a job should have, to attain the minimum cost in Theorem \ref{thm: opt_queue} while satisfying average delay constraint.
\begin{theorem}\label{thm: opt_prob}
The optimal maximal wait-time distribution for the single slot scheduling policy is given by the solution to the following linear optimization problem,
\begin{align*}
    \max_{\substack{\int_{0}^\infty f_X(w) dw = 1 \\ \int_{0}^\infty f_X(w) \left(\int_{0}^w G_\mu (y) dy \right) dw = \frac{\delta}{1-\lambda \delta}}} \int_{0}^\infty f_X(w) F_\mu(w) dw
\end{align*}
where $f_X$ is the density of the maximal wait-time $X$ and $F_\mu(w) = 1 - G_\mu(w)$ is the CDF of spot inter-arrival times with mean $1/\mu$.
\end{theorem}
To see that this is an LP, observe that the constraints and the objective are weighted sums of the density $f_X$, in the form of integral. Below is the proof of the theorem.

\begin{proof}
Little's Law applies, which is
\begin{align}
    \E[N] = \lambda \E[T]
\end{align}
where $\lambda$ is the average rate of job arrival, $T$ is the time spent by a job in the system and $N$ ($0$ or $1$) is the number of jobs in the system. For general processes, the memory-less property of Poisson arrivals is not applicable. We, therefore, guard against the worst-case: a job arriving right after a spot departs. Hence, the cost of a job can be expressed as the following random variable with the inequality,
\begin{align}
\begin{split}
    C \leq& \text{ } k \cdot N + k\cdot (1-N)\cdot \mathbbm{1}(X\leq S_\mu)\\
    &+ 1\cdot (1-N)\cdot \mathbbm{1}(X>S_\mu)
\end{split}\\
    =& \text{ } k - (k-1)\cdot (1-N)\cdot \mathbbm{1}(X>S_\mu)
\end{align}
the expectation of which is
\begin{align}
    \E[C] &\leq k - (k-1)\cdot \E[\E[(1-N)\cdot \mathbbm{1}(X>S_\mu) | N]]\\
    &= k - (k-1)\cdot \E[(1-N)\mathbbm{P}(X>S_\mu)]\\
    &= k - (k-1)\cdot (1-\E[N])\cdot \p(X>S_\mu)
\end{align}
Note that $\E[S_\mu] = \frac{1}{\mu}.$ Now, use Little's Law to replace $\E[N]$ with $\lambda \E[T]$ and $\E[T]\leq \delta$. This gives us
\begin{align}\label{eqn5}
    \E[C] \leq k - (k-1)\cdot (1-\lambda \delta)\cdot \p(X>S_\mu)
\end{align}

Now, the problem boils down to choosing the distribution for $X$ that maximizes $\p(X>S_\mu)$ while ensuring $\E[T]\leq \delta.$ The latter condition can be simplified by using the fact that
\begin{align}
    T \leq (1-N)\cdot W
\end{align}
where $W = \min\{X,S_\mu\}.$ This is because the maximum time a job will have to wait is when it joins (meaning $N=0$) and just misses the previous spot instance. Taking expectation on both sides and then applying Little's Law, gives
\begin{align}
    \E[T] \leq (1-\lambda \E[T])\E[W]
\end{align}
which can be rewritten as
\begin{align}
    \E[T] \leq \frac{\E[W]}{1+\lambda \E[W]}
\end{align}
The deadline constraint $\E[T] \leq \delta$ is, therefore, ensured by
\begin{align}
    \E[W] = \frac{\delta}{1-\lambda \delta}.
\end{align}
Now,
\begin{align}
     \p(X>S_\mu) &= \int_{0}^\infty \p(X > w) f_\mu(w) dw\\
    \E[W] = \frac{\delta}{1-\lambda \delta} &= \int_{0}^\infty G_X(w) G_\mu(w) dw
\end{align}
where $G_X(w) = 1 - F_X(w)$ and $G_\mu(w) = 1-F_\mu(w)$. Now, recall that minimizing cost per job is equivalent to maximizing $\p(X>S_\mu)$. Therefore, the optimization problem
\begin{align}
	\max_{\E[W] = \frac{\delta}{1-\lambda \delta}} \p(X>S_\mu)
\end{align}
can be formulated as
\begin{align}
    \max_{\substack{G_X: \R \to [0,1]\\ G_X' \leq 0\\ \int_{0}^\infty G_X(w) G_\mu(w) dw = \frac{\delta}{1-\lambda \delta}}} \int_{0}^\infty G_X(w) f_\mu (w) dw 
\end{align}
Apply by-parts to the integral in the objective,
\begin{align}
\begin{split}
    \int_{0}^\infty &G_X(w) f_\mu (w) dw\\
    &= \bigg[G_X(w) F_\mu(w)\bigg]_{w=0}^{w = \infty} + \int_{0}^\infty f_X(w) F_\mu(w) dw
\end{split}\\
    &= 0 + \int_{0}^\infty f_X(w) F_\mu(w) dw.
\end{align}
Now apply by-parts to the constraint,
\begin{align}
\begin{split}
    \int_{0}^\infty &G_X(w) G_\mu(w) dw \\
    =& \bigg[ G_X(w) \int_{0}^w G_\mu (w) dw \bigg]_{w=0}^{w = \infty}\\
    &+ \int_{0}^\infty f_X(w) \left(\int_{0}^w G_\mu (y) dy \right) dw
\end{split}\\
    =& \text{ } 0 + \int_{0}^\infty f_X(w) \left(\int_{0}^w G_\mu (y) dy \right) dw.
\end{align}
Therefore, the optimization problem is
\begin{align}
    \max_{\substack{\int_{0}^\infty f_X(w) dw = 1 \\ \int_{0}^\infty f_X(w) \left(\int_{0}^w G_\mu (y) dy \right) dw = \frac{\delta}{1-\lambda \delta}}} \int_{0}^\infty f_X(w) F_\mu(w) dw
\end{align}
\end{proof}

\subsection{Solving for common distributions}\label{subsec: solving_for_distrb}
In this section, we analyze the optimization problem in Theorem \ref{thm: opt_prob} by considering some common distributions for inter-arrival times of spot instances. Our first result in this is for the entire class of finite support distributions.
\begin{corollary}\label{corr: fnt_supp_srvc}
Consider a general job arrival process, with average arrival rate $\lambda$, and allowing queue length at most one in steady state.  If the spot inter-arrival time has \textbf{finite support} $[0,L]$, a maximal wait-time distribution satisfying
\begin{align*}
    \p(X\geq L) &= \frac{\mu \delta}{1-\lambda \delta}\\
    \p(X=0) &= 1-\p(X\geq L)
\end{align*}
gives the optimal cost in Theorem \ref{thm: opt_queue}.
\end{corollary}

The most interesting aspect of this result is that the solution to the optimization problem in Theorem \ref{thm: opt_prob} is agnostic to the specifics of the spot inter-arrival time. Below is the proof.

\begin{proof}
When spot inter-arrival times are finite support,
\begin{align}
    \int_0^w G_\mu(y) dy = \begin{cases}
        h(w) & w \in w\in [0,L]\\
        \frac{1}{\mu} & w>L
    \end{cases}
\end{align}
where $h(w):[0,L] \to [0,1/\mu]$ is an increasing function. The wait-time constraint integral can therefore be split into
\begin{align}
    \frac{\delta}{1-\lambda \delta} &= \int_0^\infty f_X(w) \int_0^w G_\mu(y) dy\\
    &=\int_0^L f_X(w) \int_0^w G_\mu(y) dy + \frac{1}{\mu} (1-F_X(L)) \label{eqn3}
\end{align}
which gives
\begin{align}\label{eqn2}
    1-F_X(L) = \frac{\mu \delta}{1-\lambda \delta} - \int_0^L f_X(w) \mu \int_0^w G_\mu(y) dy.
\end{align}
Now the objective to maximize evaluates as
\begin{align}
    \int_0^\infty f_X(w) F_\mu (w) dw = \int_0^L f_X(w)F_\mu(w) dw + (1-F_X(L)).
\end{align}
Plugging $(1-F_X(L))$ from \eqref{eqn2}, we get
\begin{align}\label{eqn4}
\begin{split}
     &\int_0^\infty f_X(w) F_\mu (w) dw\\
     &= \int_0^L f_X(w) \left(F_\mu(w) - \mu \int_0^w G_\mu(y) dy \right) dw + \frac{\mu \delta}{1-\lambda \delta}.
\end{split}
\end{align}
Recall Theorem \ref{thm: opt_queue}, the cost per job is lower bounded as, 
\begin{align}\label{eqn10}
	\E[C] \geq k - (k-1)\mu \delta
\end{align}
for single-slot queue. Also recall that the LHS of \eqref{eqn4} is $\p(X>S_\mu)$. Plugging this in \eqref{eqn5} gives us the following
\begin{align}
	\begin{split}
		\E[C] \leq& k - (k-1)\mu \delta -(k-1)(1-\lambda \delta) \times \ldots\\
		& \ldots \times \int_0^L f_X(w) \left(F_\mu(w) - \mu \int_0^w G_\mu(y) dy \right) dw
	\end{split}
\end{align}
First, this means the integral has to be non-positive, otherwise it contradicts \eqref{eqn10}. Second, making this integral zero gives the cost upper bound same as its hard lower bound. We, therefore, conclude that the \textit{optimal distribution is the one which makes the integral in \eqref{eqn4} zero}. Now, observe that
\begin{align}
    F_\mu(w) - \mu \int_0^w G_\mu(y) dy = 0 \text{ if } w=0,L.
\end{align}
For $w=0$, it is natural and for $w=L$ the integral in $G_\mu$ becomes $\frac{1}{\mu}$ and $F_\mu(L) = 1$. This means weights on $w=0$ and in $[L,\infty)$ allows optimal cost. Now, we come back to \eqref{eqn3}, which we can rewrite as
\begin{align}\label{eqn11}
    \frac{\delta}{1-\lambda \delta} = \int_{0}^{L^{-}} f_X(w)\int_0^w G_\mu(y) dy + \frac{1}{\mu} (1-F_X(L))
\end{align}
because $\int_0^L G_\mu(y) dy = \frac{1}{\mu}$, meaning a mass at $w = L$ has the same effect on the wait-time as the measure in $(L,\infty)$. Putting $\p(X\geq L) = p$ and $\p(X=0) = 1-p$ and plugging it in \eqref{eqn11} gives
\begin{align}
	\frac{\delta}{1-\lambda \delta} = 0 + \frac{p}{\mu} \implies p = \frac{\mu \delta}{1-\lambda \delta}
\end{align}
finishing the proof.
\end{proof}


\begin{corollary}\label{corr: uniform_srvc}
The class of distributions in Corollary \ref{corr: fnt_supp_srvc} is \textbf{unique} when the service time is uniform distribution.
\end{corollary}
\begin{proof}
We have $G_\mu(y) = 1-\frac{y}{L}$. Now consider \\ $\left(F_\mu(w) - \mu \int_0^w G_\mu(y) dy \right)$ in \eqref{eqn4} in this context,
\begin{align}
    F_\mu(w) - \mu \int_0^w G_\mu(y) dy = \frac{w^2}{L^2} - \frac{w}{L}
\end{align}
If $X$ has a positive measure in $(0,L)$, 
\begin{align}
    \int_0^L f_X(w)\left(\frac{w^2}{L^2} - \frac{w}{L} \right) dw < 0
\end{align}
and the objective is strictly less than $\frac{\mu \delta}{1-\lambda \delta}$. Consequently, the cost per job is larger than $k - (k-1)\mu \delta$ (where $\mu = 2/L$). This is not desirable. Therefore, $X$ can have positive measure only at $0$ and in $[L,\infty)$, proving uniqueness in this case.
\end{proof}
\begin{remark}\label{rem: uniform_srvc_max}
Consider the conditions in Corollary \ref{corr: fnt_supp_srvc}. If we additionally want to minimize the largest possible wait-time value, the distribution is unique
\begin{align*}
    X = \begin{cases}
        0 & \text{w.p } 1- \frac{\mu \delta}{1-\lambda \delta}\\
        L & \text{w.p } \frac{\mu \delta}{1-\lambda \delta}
    \end{cases}
\end{align*}
\end{remark}
Next, we consider exponential inter-arrival time for spot instances, a common model in job-scheduling \cite{brown2005statistical}.
\begin{corollary}\label{corr: exp_srvc}
Consider a general job arrival process, with average rate $\lambda$, and allowing queue length at most one in steady state. Any maximal wait-time distribution satisfying
\begin{align*}
    \mathcal{L}\{f_X\}(\mu) = \frac{1 - (\lambda + \mu)\delta}{1-\lambda \delta}
\end{align*}
has optimal cost when $S_\mu \sim Exp(\mu)$. Here $\mathcal{L}\{f_X\}(\cdot)$ is the Laplace transform of density $f_X$ of the maximal wait-time.
\end{corollary}
\begin{proof}
$G_\mu(x) = 1 - F_\mu(x) = \exp(-\lambda x)$, implying
\begin{align}
\begin{split}
    &\int_{0}^\infty f_X(w) \left(\int_{0}^w G_\mu (y) dy \right) dw \\
    &= \int_{0}^\infty f_X(w) \left(\frac{1-\exp(-\mu w)}{\mu} \right) dw
\end{split}\\
&= \frac{1}{\mu}\left(\int_{0}^\infty f_X(w) dw - \int_{0}^\infty f_X(w) \exp(-\mu w) dw \right)\\
&= \frac{1 - \mathcal{L}\{f_X\}(\mu)}{\mu} = \frac{\delta}{1-\lambda \delta}
\end{align}
which gives the condition in the corollary. As for why it alone gives optimal cost, the value of the objective proves that.
\begin{align}
	\int_{0}^\infty f_X(w) F_\mu(w) dw &= \int_{0}^\infty f_X(w) (1-\exp(-\mu w)) dw\\
	&= 1 - \mathcal{L}\{f_X\}(\mu) = \frac{\mu \delta}{1-\lambda \delta}.
\end{align}
\end{proof}

\begin{remark}\label{rem: exp_srvc_exp_reneg}
Consider the conditions of Corollary \ref{corr: exp_srvc}. If $X \sim Exp(\phi)$, the optimal rate is $\phi = \frac{1}{\delta} - (\mu + \lambda).$
\end{remark}

\begin{corollary}\label{corr:exp_srvc_max}
Consider the conditions in Corollary \ref{corr: exp_srvc}. If we additionally want to minimize the largest possible wait-time value, the distribution is unique and deterministic,
\begin{align*}
    X=  \frac{1}{\mu} \log \left[\frac{1-\lambda \delta}{1- (\lambda+ \mu) \delta} \right]
\end{align*}
\end{corollary}
\begin{proof}
Consider any density function $f_X$. Using Jensen's inequality, we have that
\begin{align}\label{eqn12}
    \int f_X(w) \exp (-\mu w) dw &\geq \exp \left( -\mu \int w f_X(w) dw \right)\\
    \frac{1 - (\lambda + \mu)\delta}{1-\lambda \delta} &\geq \exp (-\mu \E[X])\\
    \implies \E[X] &\geq \frac{1}{\mu} \log \left[\frac{1-\lambda \delta}{1- (\lambda+ \mu) \delta} \right]\label{eqn13}
\end{align}
and $\max X \geq \E[X]$. Therefore, the largest wait-time is at least the RHS of the above equation. This minimum value is achieved when Jensen's \eqref{eqn12} is an equality and  $\max X = \E[X]$, happening only when $X$ is deterministic and equal to $\frac{1}{\mu} \log \left[\frac{1-\lambda \delta}{1- (\lambda+ \mu) \delta} \right]$.
\end{proof}

\section{Scheduling under high target delay $\delta$}
In the previous sections, we performed an in-depth analysis for the case where delay constraint is strong ($\delta$ is small). In this section, we use the insights from the last section to address the case when the mean delay constraint is larger. In such cases, the following questions arise:
\begin{enumerate}
    \item Can a larger (than one) queue-length be allowed?
    \item What should be the joining/waiting policy for these jobs?
    \item Is there be a cap on queue-length $N$ to prevent overload?
\end{enumerate}
To answer all these questions, we go back to Theorem \ref{thm: gen_cost} and Little's Law, to first formulate this general problem. Recall that for any $G/G/1$ system, we are trying to minimize $\E[C] = k - (k-1)\cdot \frac{\mu}{\lambda} (1-\pi_0)$, which amounts to maximizing $(1-\pi_0) = \sum_{n=1}^\infty \pi_n.$ Now, Little's law dictates that $\lambda \delta = \E[N] = \sum_{n=1}^\infty n \pi_n.$ With a general job-arrival process, general spot inter-arrival times and any scheduling policy, the objective is:
\begin{align}\label{eqn: gen_opt_problem}
    \max_{\substack{ 0\leq \pi_n \leq 1}} \quad\sum_{n=1}^\infty \pi_n
\end{align}
subject to constraints
\vspace{-5pt}
\begin{align}
    \sum_{n=1}^\infty n \pi_n &\leq \lambda\delta \label{eqn: wait-time-constraint}\\
    \sum_{n=1}^\infty \pi_n &\leq 1 \label{eqn: total-probb-constraint}
\end{align}

\subsection{Near-optimal Heuristic Scheduling Policy}\label{subsec: heurstic-scheduling}
For $G/G/1$ system, one cannot characterize the steady state distribution $\{\pi_n\}_n$ or the probability generating function $p(z)$ in closed form. Therefore, both the wait-time constraint and the objective cannot be expressed as a function of the decision variables: (i) joining probabilities $\{q_n\}_{N=0}^\infty$ based on the observed queue-length $N$ and, (ii) the wait-time distribution $f_{X_n}(\cdot)$ for the $n^{th}$ job in the queue. 

We, therefore, exploit the structure involved in the optimization problem \eqref{eqn: gen_opt_problem} to come up with a greedy scheduling policy, without explicitly using the dependence on steady state probabilities $\{\pi_n\}_n$ on the decision variables. The following result summarizes this greedy approach:
\begin{theorem}\label{thm: gen_sys_greedy}
    Consider general job arrival process, general spot arrival process and average wait-time constraint $\delta \in \left(\frac{N^*}{\lambda},\frac{N^*+1}{\lambda} \right]$. In steady state, greedy scheduling policy has the following 3-phase form:
    \begin{enumerate}
        \item Phase 1: Jobs join and wait indefinitely: 
        \begin{align*}
            \exists \Hat{N} > N^* \text{ such that } q_N &= 1 \text{ } \forall \text{ } N<\Hat{N}\\
            X_n &= \infty \text{ } \forall \text{ } n\leq \Hat{N}
        \end{align*}
        \item Phase 2: For the $(\Hat{N}+1)^{(th)}$ job, it will have a joining probability $q_{\Hat{N}+1} \in (0,1)$ such that with wait-time $X_{\Hat{N}+1} = \infty$, the average wait-time of the system is exactly $\delta$.
        \item Phase 3: Arriving jobs are go directly on to on-demand: $$q_N = 0 \text{ } \forall \text{ } N\geq\Hat{N}+1.$$
    \end{enumerate}
\end{theorem}
The goal is to come up with a scheduling policy and prove its optimality from a rather abstract LP problem \eqref{eqn: gen_opt_problem},\eqref{eqn: wait-time-constraint}, \eqref{eqn: total-probb-constraint} which might not even admit a real policy.

\begin{proof}
The first idea is to keep in mind the knapsack structure of the objective \eqref{eqn: gen_opt_problem} with \eqref{eqn: wait-time-constraint} and \eqref{eqn: total-probb-constraint}. Notice that all $\pi_n$ contribute equally to the objective and the constraint \eqref{eqn: total-probb-constraint}. On the other hand, in \eqref{eqn: wait-time-constraint}, the corresponding coefficients are increasing in $n$. This signals that one needs to prioritize maximizing $\pi_n$ for \textit{lower} $n$.

Second, we have $\lambda\delta \in (N^*,N^*+1].$ This motivates us to split the wait-time constraint into two sums:
$$\sum_{n=1}^{N^*} n\pi_n + \sum_{n=N^*+1}^\infty n\pi_n$$
which create two phases for the formulating the scheduling policy. Up to queue-length $N^*$, the scheduler have any policy and still satisfy $\sum_{n=1}^{N^*} n\pi_n < N^* < \lambda\delta.$ Hence, to maximize spot-instance utilization, the optimal policy will at least satisfy:
\begin{align}
    q_{N-1} = 1 \text{ } \forall \text{ } N<N^* \text{ and }
    X_n = \infty \text{ } \forall \text{ } n\leq N^*.
\end{align}
Now, suppose we follow this policy with additionally capping the queue-length at $N^*$ and define the average wait-time under this policy as $\Bar{T}_{N^*}$. Then the slack in wait-time constraint will be $\delta-\Bar{T}_{N^*}$. 

Third, to exploit this slack, we need to schedule more jobs into spot-processing. This has to be done in a manner to maximize spot-instance utilization while ensuring maximization of the slack delay-constraint. Recall that this constraint suggests we prioritize maximizing for lower index $n$ first. Consequently, the following steps can be iteratively performed right until the delay constraint is just violated, starting with $k=1$:
\begin{enumerate}[label=(\alph*)]
    \item Set joining probabilities as
    \begin{align}
        q_N = \begin{cases}
            1 & N\leq N^*+k\\
            0 & N>N^*+k
        \end{cases}
    \end{align}
    and $X_n=\infty$ for all $n\leq N^*+k$.
    \item Empirically calculate the average wait-time $\Bar{T}_{N^*+k}$
    \item Stop if $\delta < \sum_{n=1}^{N^*+k}n \pi_n^{(k)},$ else increment $k$ by $1$ and repeat.
\end{enumerate}
Suppose we stop at $k = \Hat{k}$ and define $\hat{N} := N^*+\Hat{k}-1.$ This means any job can freely join and wait indefinitely until a queue-length of $\Hat{N}$ is reached. For the $(\Hat{N}+1)^{(th)}$ incoming job, if it joins with probability $q_{\Hat{N}} = 1$ and waits indefinitely, the average wait-time constraint is violated, that is $\Bar{T}_{\Hat{N}+1} > \delta$. At the same time, if it immediately goes for on-demand processing, the delay constraint is not completely utilized. 


Lastly, for the $(\Hat{N}+1)^{(th)}$ job, we have to exploit the remaining slack by one of two methods: (i) set $q_{\hat{N}} = 1$ and optimize over an infinite-dimensional distribution $f_{X_{\Hat{N}+1}}(\cdot)$ or, (ii) Set $X_{\Hat{N}+1} = \infty$ and tune $q_{\Hat{N}}$ in $(0,1)$. It is clear that the latter option is the only tractable method. Observe that tuning $q_{\Hat{N}}$ in $(0,1)$ will in fact lead to $\E[T] = \delta$ (binding and not slack) because for $q_{\Hat{N}} = 0$ there is slack and for $q_{\Hat{N}} = 1$, there is violation.

Having exploited all slack in the delay constraint, any more jobs arriving will naturally be sent directly to on-demand instances, that is, 
$$q_N = 0 \text{ } \forall \text{ } N\geq\Hat{N}+1.$$
\end{proof}
\vspace{-10pt}
Finally, for finite-support spot inter-arrival times, we show this technique is in fact optimal.
\begin{remark}
    The effective delay constraint for the $(\Hat{N}+1)^{(th)}$ job is small (due to jobs in front waiting indefinitely). Recall Corollary \ref{corr: fnt_supp_srvc}, where over finite support $S_\mu \in [0,L]$, the optimal policy (for small $\delta$) is to have a single job waiting with maximal wait-time density $f_X$ having masses only at $\{0\}$ and in $[L,\infty)$. Such a distribution is exactly equivalent to joining the queue with probability $q_{\Hat{N}+1}$ =  $\p(X\geq L)$ and waiting indefinitely for spot-service. 
\end{remark}



\subsection{Learning optimal scheduling on-the-fly}
The scheduling policy laid out by Theorem \ref{thm: gen_sys_greedy} is aimed at maximizing spot instance utilization (by having $X_n = \infty$ for $n\leq \Hat{N}+1$) and optimizing over the knapsack structure of the abstract optimization problem \eqref{eqn: gen_opt_problem},\eqref{eqn: wait-time-constraint},\eqref{eqn: total-probb-constraint}. Since both of these objectives directly correlate to minimizing average cost per job, the only objective now is to exploit the slack in the delay constraint to the fullest.

Such a policy is required to compute the optimal $\hat{N}$ and the joining probability $q_{\hat{N}}$. $\Hat{N}$ along with $q_{\hat{N}}$ represents a system that has a fractional maximum queue-length $r^* \in (\hat{N},\Hat{N}+1)$. We model this continuous variable as 
\begin{align}\label{eqn: r_formula}
    r^* = \Hat{N}+q_{\Hat{N}}.
\end{align}
Under this modeling assumption, we want to learn the optimal $r^*$ and then utilize the scheduling policy in Theorem \ref{thm: gen_sys_greedy} with $\Hat{N} = \lfloor r^* \rfloor$ and $q_{\Hat{N}} = r^* - \Hat{N}.$ 

We had noted that one might need to compute the average wait-time for each system with maximum queue-length in $\{N^*+1,\ldots,\Hat{N}\}$ and then tune $q_{\hat{N}} \in (0,1)$ to exploit the wait-time constraint to the fullest. With the model \eqref{eqn: r_formula}, we can optimize over these two quantities simultaneously in the following manner. 

To converge to zero slack in the average-delay, while ensuring average cost minimization, we use the following methodology:
\begin{enumerate}
    \item Set the scheduling policy as the one in Theorem \ref{thm: gen_sys_greedy} with $\Hat{N} = \lfloor r \rfloor$ and $q_{\Hat{N}} = r-N$
    \item Apply this policy for jobs/spot arrivals over a window $W$ to empirically estimate the average delay $d(r)$
    \item Penalize the slack/violation in the delay constraint through the following loss function:
    \begin{align}
        L(r) = \frac{1}{2}(d(r)-\delta)^2
    \end{align}
    and use gradient descent on it to update $r$.
\end{enumerate}



Algorithm \ref{alg: adaptive_adm_ctrl} below summarizes and implementation of the aforementioned technique. It is worthwhile to note that our algorithm \textit{learns} the optimal parameters while scheduling jobs according to the greedy-optimal policy (Theorem \ref{thm: gen_sys_greedy}). 

\begin{algorithm}
\caption{Adaptive Admission Control Policy}
\label{alg: adaptive_adm_ctrl}
\begin{algorithmic}[1]
\REQUIRE Target average delay $\delta$, step-size $\alpha > 0$, initial $r \in [0, r_{\max}]$, averaging window $W$, error tolerance $\epsilon$
\ENSURE Admission control: $r = N + p$ (learning objective) 

\STATE Initialize $r \gets r_0 \in [0, r_{\max}]$
\WHILE{$|\Hat{d}-\delta| < \epsilon$}
    \STATE Set $N \gets \lfloor r \rfloor$, $p \gets r - N$
    \STATE Set the scheduling policy:
    \begin{itemize}
        \item If queue length $< N$: admit job
        \item If queue length $= N$: admit with probability $p$
        \item If queue length $> N$: reject and use on-demand
    \end{itemize}
    \STATE Observe arrivals \& completions over time window $W$
    \STATE $d(r) \gets$ empirical avg. delay  for jobs arriving in $W$
    \STATE Compute gradient $\eta \gets \eta \cdot (d(r) - \delta)$

    \STATE Update: $r \gets \min\left( r_{\max}, \max(0, r - \eta) \right)$
\ENDWHILE
\end{algorithmic}
\end{algorithm}
The optimality of the greedy approach in Theorem \ref{thm: gen_sys_greedy} and the near-optimality of Algorithm \ref{alg: adaptive_adm_ctrl}'s learning-with-scheduling scheme is cemented by extensive numerical experiments.

\section{Experiments}\label{section:experiments}

Our experimental set-up is as follows. Owing to the generality of the job arrival process, we test under two distributions for the job inter-arrival times: (i) $\exp{(1/\lambda)}$ and (ii) Gamma$(1/\lambda,1)$. For the spot availability process, we consider two settings: (i) statistically-modeled instance availability process by \cite{kadupitige2020modeling} and (ii) $Poisson(\mu)$. On-demand instances carry a cost of $k=10$ while spot instances are unit cost. Average wait-times/delays are in the unit of hours.

For each pair of job arrival and spot availability process, we run Algorithm \ref{alg: adaptive_adm_ctrl} from a small and large start point $r_{init}$ to understand the average-cost trajectory and the average-delay trajectory as $r$ converges to the optimal \textit{proxy-}max-queue-length $r^*$. We plot the running-average of cost per job $C(r(n))$, delay per job $d(r(n))$. For each such plot, we also provide an \textit{inset} plot tracking the evolution (and convergence) of $r(n)$ as Algorithm \ref{alg: adaptive_adm_ctrl} runs through incoming job arrivals.

\subsection{Preemptible GCP instances}
As discussed in detail in Section \ref{sec: related works}, spot instance availability data is scarce due to the vendors revealing either qualitative availability information, like AWS interruption rate $ \in \{1,2,3\}$ (see \cite{azure_spots,AWS_EC2}) or none as in the case of Google Cloud Platform (GCP) VMs \cite{GCP,azure_spots}. 

\cite{kadupitige2020modeling} characterizes the statistics of spot availability by studying the life-time of preemptible GCP VMs, with evidence found in \cite{li2020characterizing}. Empirically, they find spot instances to be available either almost immediately or closer to the $\approx24$ hour preemption deadline set by Google. The density function for inter-arrival times $f_S(\cdot)$ for such spot instances follows a \textit{bath-tub} function, which they model using the following cumulative-distribution function:
\begin{align}
    F_{S}(t) = A\left(1 - \exp{\left(-\frac{t}{\tau_1}\right)} + \exp{\left(\frac{t-b}{\tau_2}\right)}\mathbbm{1}_{t\leq \tau_2} \right)
\end{align}
where they find parameter values $b\approx 24$ hours, $\tau_1 \in [0.5,1.5]$, $\tau_2 \approx 0.8$ and $A \in [0.4,0.5]$ to work well in practice. We sample spot inter-arrival times from the following distribution and consider job arrivals as a $Poisson(1/12)$ or Gamma$(1/12,1)$ process ($\approx 12$ hours inter-arrival times in both processes).

Figure \ref{fig: gcp_low_delta} shows the performance of Algorithm \ref{alg: adaptive_adm_ctrl} as $r^* \approx 0.3$ is learned from a low or high $r_{init}$ when $\lambda = \frac{1}{12}$, $\mu \approx \frac{1}{12}$ and delay constraint $\delta = 3$ hours. Both trajectories converge to the single-slot optimal cost: $k-(k-1)\mu\delta \approx 7.75$. 

\begin{figure}
    \centering
    \includegraphics[width=0.95\linewidth]{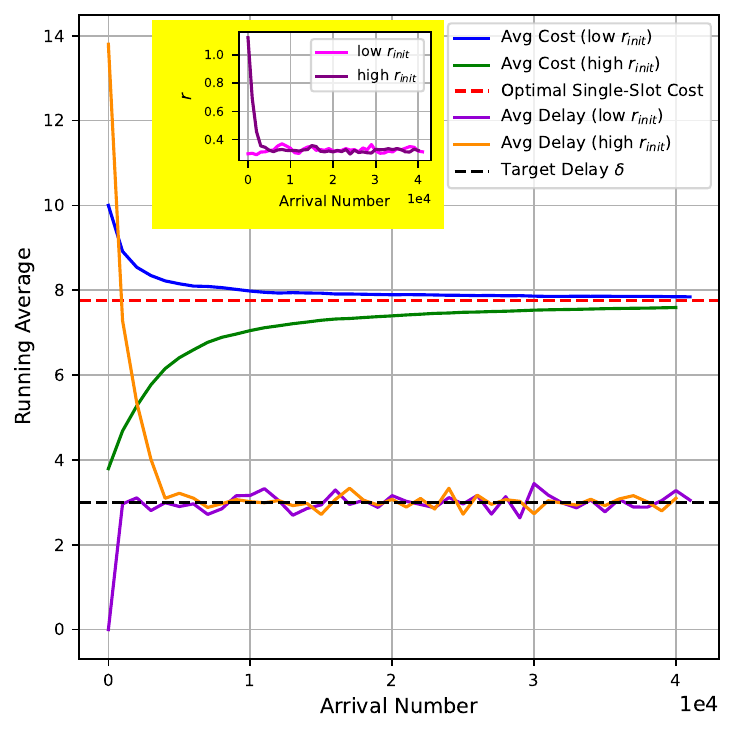}
    \caption{Learning optimal scheduling $\delta < \frac{1}{\lambda + \mu}$}
    \label{fig: gcp_low_delta}
    \vspace{-19pt}
\end{figure}

This empirical result corroborates multiple theoretical claims made in previous sections. First, it validates Theorem \ref{thm: opt_queue} that capping queue length to $1$ is optimal for small $\delta$ and that the optimal cost is given by $k-(k-1)\mu\delta$. Second, it shows that the greedy scheduling policy lined out in Theorem \ref{thm: gen_sys_greedy} is optimal in minimizing the average-cost per job in the strong delay constraint regime, as seen by cost-curves converging to the single-slot optimal cost. Lastly, the learning-plus-scheduling technique of Algorithm \ref{alg: adaptive_adm_ctrl} successfully implements this greedy policy in an \textit{online} fashion to learn the optimal $r^*$. 

Fig. \ref{fig: gcp_high_delta} illustrates Algorithm \ref{alg: adaptive_adm_ctrl} in the relaxed delay constraint regime with $\delta = 18$ hours. Although the theoretical optimal cost cannot be analytically computed, as noted in Section \ref{subsec: heurstic-scheduling}, both cost-curves still converge to a common value while the algorithm is learning to satisfy the delay constraint.

\subsection{$M/M/1/N$ system}
We now test the optimality of Theorem \ref{thm: gen_sys_greedy}'s scheduling policy and Algorithm \ref{alg: adaptive_adm_ctrl}'s learning technique for both small and large $\delta$ values. To do that, we need the optimal average job cost in an analytical form, which can be formulated only for an $M/M/1$ system. Despite the variability in real-world behavior, the exponential distribution has been widely used in modeling preemptions for systems such as AWS EC2 spot instances \cite{sharma2016flint,sharma2016novel,zheng2015bid}.
To that end, we consider a $Poisson(1/12)$ job-arrival process and $Poisson(1/24)$ spot-instance arrival process with $\delta=3$ hours and $\delta=27$ hours.

\begin{figure}
    \centering
    \includegraphics[width=0.95\linewidth]{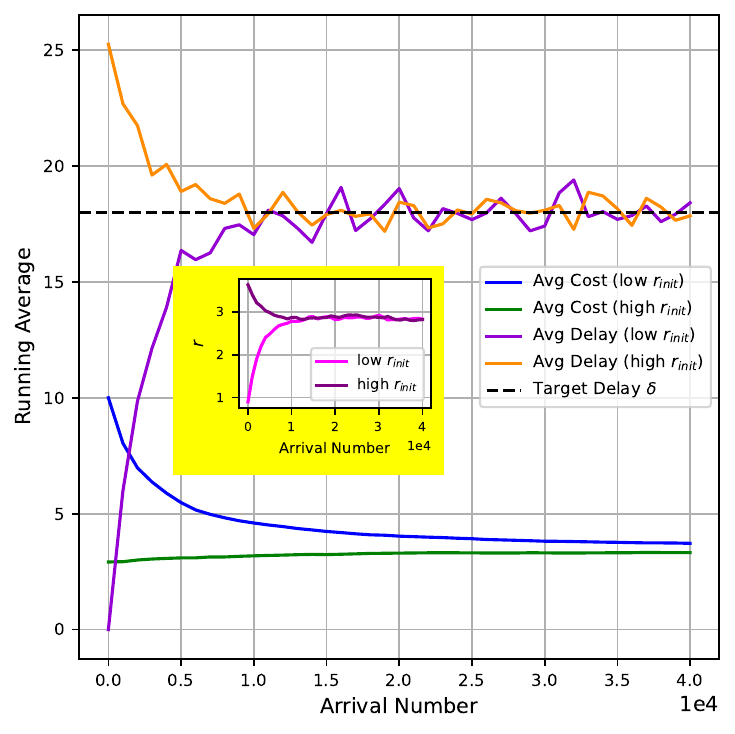}
    \caption{Learning optimal scheduling for $\lambda \delta > 1$}
    \label{fig: gcp_high_delta}
    \vspace{-15pt}
\end{figure}

Figure \ref{fig: mm1 low delta} validates the optimality of our scheduling policy and $r^*$-learning algorithm in the strong delay constraint regime where the cost-curves converge to $\E[C^*] = k-(k-1)\mu\delta = 8.875$ when the delay curves converge to $\delta = 3$. 

For the large $\delta$ case, we first need to formulate the optimal cost for an $M/M/1/N$ queue for spot-instance processing, as presented in the theorem below. What the this guarantees is that when one has average delay constraint $\delta = \delta_N$, the optimal policy is to have allow a queue-length of at most $N$. 

\begin{figure}
    \centering
    \includegraphics[width=0.95\linewidth]{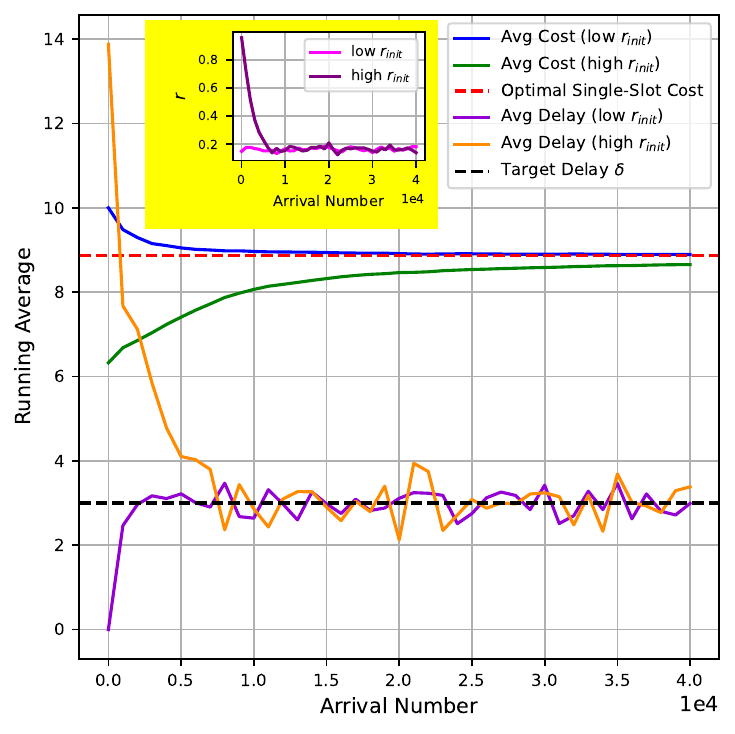}
    \caption{Learning optimal scheduling $\delta < \frac{1}{\lambda + \mu}$ for memoryless system}
    \label{fig: mm1 low delta}
    \vspace{-10pt}
\end{figure}

\begin{figure}
    \centering
    \includegraphics[width=0.95\linewidth]{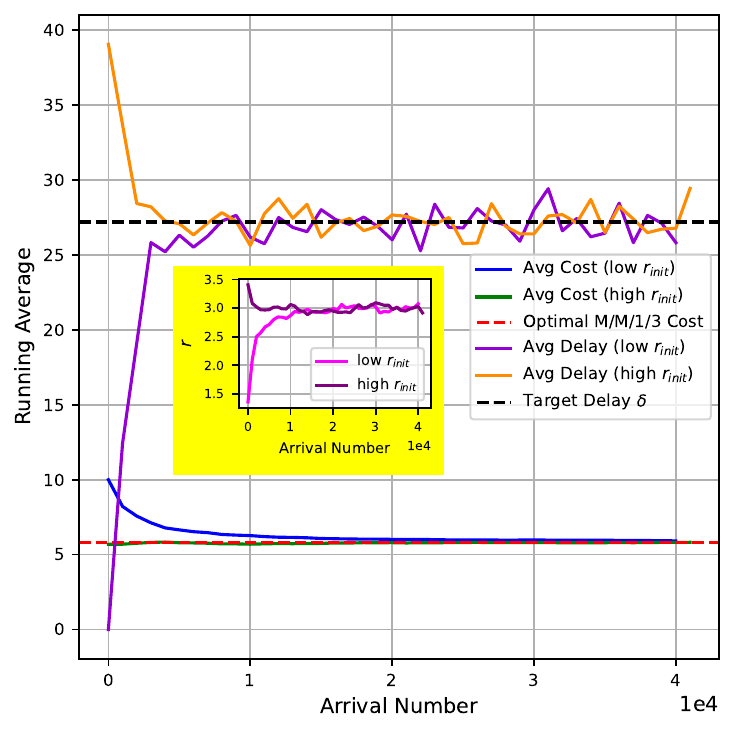}
    \caption{Learning optimal scheduling $\lambda\delta>1$ for memoryless system}
    \label{fig: mm1 large delta}
    \vspace{-18pt}
\end{figure}

\begin{theorem}\label{thm: mm1-m}
    Consider the job arrival process to be $Poisson(\lambda)$, the spot inter-arrival times to be IID $Exp(\mu)$ and the maximal wait-time distribution to be $Exp(\alpha)$, $\alpha > 0$. The minimum possible cost if the allowed queue length is at most $N$, under any joining policy $\{q_n\}_{n=0}^N$, is 
    \vspace{-5pt}
    \begin{align*}
        \E[C_N] &= k - (k-1)\left(\frac{\mu}{\lambda} \right) \left(1-\frac{\frac{\lambda}{\mu}-1}{\frac{\lambda^{N+1}}{\mu^{N+1}}-1} \right)
    \end{align*}
    which is a strictly decreasing function in $N$ under light traffic ($\lambda < \mu$) or overload ($\lambda > \mu$). The delay $\delta_N$ that allows this minimum cost to be achieved must satisfy,
    \vspace{-5pt}
    \begin{align*}
        \delta_N \geq \frac{1}{\lambda}\sum_{n=1}^N \frac{n \left(\frac{\lambda}{\mu}\right)^n}{1+\sum_{n=1}^N \left(\frac{\lambda}{\mu}\right)^n}
    \end{align*}
    which is a strictly increasing function in $N$.
\end{theorem}
For $N=3$, $\delta_N\approx 27$, meaning $\lambda\delta>1$ and that we are in the relaxed delay constraint regime. However, Algorithm \ref{alg: adaptive_adm_ctrl} does not assume knowledge of $\lambda$ or $\mu$ before-hand to analytically compute this $N$. 

Figure \ref{fig: mm1 large delta} confirms our claim that Algorithm \ref{alg: adaptive_adm_ctrl} can learn the correct $r^* \approx N = 3$ under the scheduling policy set by Theorem \ref{thm: gen_sys_greedy} under higher target delay $\delta$. We observe optimality with respect to all three quantities. First $r$ converges to the optimal $r^* = N=3$. Second, the average cost per job converges to the optimal cost given by Theorem \ref{thm: mm1-m} for $N=3$ and lastly, the delay constraint is satisfied.

\section{Related Work}\label{sec: related works}
AWS introduced spot instances in 2009, employing a bidding system to utilize unused cloud capacity \cite{AWS_EC2}. As a result, significant research has been dedicated to understanding and analyzing the pricing strategies of these service providers \cite{agmon2013deconstructing,javadi2011statistical,wang2013present,song2012optimal,singh2015dynamic,tang2012towards}. Spot price based bidding is loosely related to the concept of scheduling with earlier works making progress \cite{menache2014demand,varshney2018autobot,poola2014fault,zafer2012optimal,song2012optimal}. However, in the presence of stable spot pricing, like AWS EC2 \cite{AWS_EC2_prices}, (Google Cloud Platform) GCP's fixed 30-day spot price \cite{GCP}, Oracle Cloud's consistent 50\% discount for preemptible instances \cite{oracle_spots}, and Azure's stable regional pricing \cite{azure_spots}, this the bidding-approach to scheduling is irrelevant. 
Recent works \cite{azure_spots,kadupitige2020modeling,10.1145/3589334.3645548} have started to focus on statistical analysis of spot-instance availability in AWS and GCP. Scheduling on spot/on-demand instances has recently garnered attention in this community with algorithms making use of ML/deep learning models to predict future spot-availability \cite{harlap2018tributary,yang2022spot,yang2023snape}, and focusing on workflow scheduling \cite{taghavi2023cost,9893047}. However, most of these works do not consider a job stream \cite{295489} or deadline constraints \cite{9640599}, with the exception of \cite{11038902} proposing a priority-scheduling algorithm for heterogeneous job arrivals, albeit for a simpler Poisson arrival model.

\section{Final Thoughts}
This work establishes a solid analytical foundation for cost-effective scheduling of delay-sensitive jobs across cloud instance types, offering optimal strategies under diverse delay constraints. The proposed adaptive algorithm demonstrates near-optimal performance, opening avenues for future research. Promising directions include extending the model to multi-tier cloud architectures, incorporating job priorities or deadlines, and exploring real-time learning-based scheduling in dynamic cloud environments.
\bibliographystyle{IEEEtran}
\bibliography{refs}
\end{document}